\documentclass[prX,onecolumn,showpacs,showkeys,amsmath,amssymb]{revtex4-2} 
\def\oper{{\mathchoice{\rm 1\mskip-4mu l}{\rm 1\mskip-4mu l}
{\rm 1\mskip-4.5mu l}{\rm 1\mskip-5mu l}}}
\def\<{\langle}
\def\>{\rangle}
\usepackage[T1]{fontenc}
\usepackage[utf8]{inputenc}
\usepackage{mathtools,amsthm}
\usepackage{dsfont,bbm}
\usepackage{bm,soul}
\usepackage[most]{tcolorbox}

\usepackage[scr=boondoxo,scrscaled=1.05]{mathalfa}
\usepackage{graphicx}
\usepackage{float}
\usepackage{subcaption}
\usepackage{xcolor}
\usepackage{enumitem}

\usepackage[normalem]{ulem}

\usepackage{hyperref}
\hypersetup{
	colorlinks = true, 
	urlcolor = darkgreen, 
	linkcolor = darkblue, 
	citecolor = darkred, 
 breaklinks = true
}
\colorlet{darkred}{red!55!black}
\colorlet{darkgreen}{green!25!black}
\colorlet{darkblue}{blue!60!black}
\newtheorem{proposition}{Proposition}
\newtheorem{remark}{Remark}
\newtheorem{example}{Example}
\newtheorem{corollary}{Corollary}

\newtheorem{thm}{Theorem}

\renewcommand{\eprint}[2][]{\href{https://arxiv.org/abs/#2}{arXiv:~\nolinkurl{#2}}} 
\newcommand{\Md}{\mathcal{M}_d(\mathbb{C})}

\begin{document}   
	
	
	\title{
 A universal constraint for relaxation rates for quantum Markov generators: \\
 complete positivity and beyond\footnote{Dedicated to Professor Ryszard Horodecki on the occasion of his 80th birthday.}  }

 \author{Dariusz Chru\'sci\'nski}
	\email{darch@fizyka.umk.pl}
	\affiliation{Institute of Physics, Faculty of Physics, Astronomy and Informatics, Nicolaus Copernicus University, Grudziadzka 5/7, 87-100 Toru\'n, Poland}

\author{Frederik vom Ende}
\email{frederik.vomende@gmail.com}
\affiliation{Dahlem Center for Complex Quantum
Systems, Freie Universit\"at Berlin, Arnimallee 14, 14195 Berlin, Germany}

 \author{Gen Kimura}
	\email{gen.kimura.quant@gmail.com}
	\affiliation{College of Systems Engineering and Science, Shibaura Institute of Technology, Saitama 330-8570, Japan}

	\author{Paolo Muratore-Ginanneschi}
	\email{paolo.muratore-ginanneschi@helsinki.fi}
	\affiliation{Department of Mathematics and Statistics, University of Helsinki PL 68, FI-00014, Finland}

	\date{\today}
	
	\begin{abstract}
Relaxation rates are key characteristics of quantum processes, as they determine how quickly a quantum system thermalizes, equilibrates, decoheres, and dissipates. While they play a crucial role in theoretical analyses, relaxation rates are also often directly accessible through experimental measurements. Recently, it was shown that for quantum processes governed by Markovian semigroups, the relaxation rates satisfy a universal constraint: the maximal rate is upper-bounded by the sum of all rates divided by the dimension of the Hilbert space. This bound, initially conjectured a few years ago, was only recently proven using classical Lyapunov theory. In this work, we present a new, purely algebraic proof of this constraint. Remarkably, our approach is not only more direct but also allows for a natural generalization beyond completely positive semigroups. We show that complete positivity can be relaxed to 2-positivity without affecting the validity of the constraint. This reveals that the bound is more subtle than previously understood: 2-positivity is necessary, but even when further relaxed to Schwarz maps, a slightly weaker---yet still non-trivial---universal constraint still holds. Finally, we explore the connection between these bounds and the number of steady states in quantum processes, uncovering a deeper structure underlying their behavior.
	\end{abstract}
	

	\maketitle
	
\section{Introduction}

The dynamics of an open quantum system is represented by a quantum-dynamical map, i.e., a family of completely positive trace-preserving (CPTP) maps $\{\Lambda_{t,0}\}_{t\geq 0}$ acting on states \cite{Davies1976,Alicki-Lendi,BrPe2002,GardinerZoller2004,LidD2019,ChrD2022,RiHu2012,BeFlo2005}. Given an initial state of the system represented by a density  operator $\bm{\rho}$ one finds the corresponding state at time $t$ by applying the map $\Lambda_{t,0}$ (quantum channel) to the initial state, i.e., $\bm{\rho}_t = \Lambda_{t,0}(\bm{\rho})$. If the system is isolated from the environment, the map $\Lambda_{t,0}$ realizes the standard unitary Schr\"odinger evolution $\Lambda_{t,0}(\bm{\rho}) = U_{t,0} \bm{\rho} U_{t,0}^\dagger$, where $U_{t,0}$ is a time-dependent unitary operator. In general, however, a quantum-dynamical map is not unitary and describes processes such as decoherence, dissipation, and relaxation that go beyond the unitary scenario \cite{BrPe2002,GardinerZoller2004,LidD2019}. In the Markovian regime, {which typically assumes weak coupling and a separation of time scales between the system and environment,} the evolution of the system's state is governed by the celebrated Markovian evolution \cite{GKLS,LinG1976}
$$
\dot{\bm{\rho}}_t = \mathfrak{L}(\bm{\rho}_t)\,,
$$
where the so-called Gorini-Kossakowski-Lindblad-Sudarshan (GKLS) generator takes the following form:
\begin{align}
	\mathfrak{L}(\bm{{\rho}}) =-\imath\,\left[\operatorname{H},\bm{\rho}\right]+
\sum_{\ell=1}^{d^{2}-1}\gamma_{\ell} \left( \operatorname{L}_{\ell}\bm{\rho} \operatorname{L}_{\ell}^\dagger - \frac 12 \{ \operatorname{L}_{\ell}^\dagger \operatorname{L}_{\ell},\bm{\rho}\} \right) \,,
		\label{LGKS}
	\end{align}
with a Hermitian effective Hamiltonian $\operatorname{H}$, noise (or jump) operators $\operatorname{L}_\ell$, and positive transition rates $\gamma_\ell \geq 0$ {for all $\ell=1,2,\ldots,d^2-1$}. Actually, positivity of $\gamma_\ell$ is necessary and sufficient to generate a semigroup of CPTP maps $\Lambda_{t,0} = e^{t \mathfrak{L}}$ \cite{GKLS,LinG1976,Alicki-Lendi}. Today, quantum Markovian semigroups describe a plethora of important processes, including quantum optical systems \cite{BrPe2002}, control of quantum information processing devices \cite{WiMi2009}, and heat transfer in solid-state quantum integrated circuits \cite{PeKa2021}. It should be stressed that the transition rates $\gamma_\ell$ are not directly observable quantities. What is often measured in the laboratory are relaxation rates $\Gamma_\ell$ that control the speed of various processes like relaxation, decoherence, equilibration, or thermalization.
Contrary to the $\gamma_\ell$'s which in principle can be arbitrarily chosen, the relaxation rates $\Gamma_\ell$ satisfy some additional constraints implied by the very requirement of complete positivity. A key example of such constraints is the celebrated relation between longitudinal $\Gamma_{\rm L}$ and transversal rates $\Gamma_{\rm T}$ for qubit systems \cite{GKLS}:
\begin{equation} \label{LT}
 2 \Gamma_{\rm T} \geq \Gamma_{\rm L} \,.
\end{equation}
More generally, it was shown \cite{KimG2002} that (\ref{LT}) provides a special example of a universal relation for qubit systems
\begin{equation} \label{Gen}
 \Gamma_{\rm max} \leq \frac 12 \sum_{\ell=1}^{3} \Gamma_\ell \,,
\end{equation}
where $\Gamma_{\rm max}$ denotes the maximal rate. Recently, it was also proven \cite{JPA2025} that for general $d$-level systems the following constraint holds true
\begin{equation} \label{CON}
 \Gamma_{\rm max} \leq \frac 1d \sum_{\ell=1}^{d^2-1} \Gamma_\ell \,,
\end{equation}
which reduces to (\ref{Gen}) for $d=2$. This was already conjectured in \cite{ChKiKoSh2021}, where, together with its proof in several special cases, it was also shown to be tight.
It was stressed \cite{KimG2002,ChKiKoSh2021,JPA2025} that (\ref{CON}) provides a physical manifestation of the very concept of complete positivity. Interestingly, the proof obtained in \cite{JPA2025} is based on
 the theory of Lyapunov exponents,
a central tool in the analysis of classical dynamical systems \cite{OseV1968,BeGaGiSt1980, SkoC2008,Pikovsky2016}, {especially} in the study of classical \cite{Devaney1986,CoEc2006,CeCeVu2009}, and quantum chaos \cite{Gutzwiller1991,Stoeckmann1999,Haake2018}, and fully developed turbulence theory \cite{BoJePaVu2005}.
In the context of quantum physics, Lyapunov exponents have also been applied to the analysis of quantum channels \cite{Daniel-1}, strongly correlated many-body systems \cite{Gharibyan2019}, and out-of-time order correlators \cite{Bergamasco2023}. In particular, \cite{Maldacena2016} conjectures the existence of a universal bound on the maximal Lyapunov exponent controlling how fast out-of-time-order correlation functions can grow. In \cite{JPA2025} it was shown that the maximal relaxation rate of the quantum generator $\mathfrak{L}$ coincides with the characteristic Lyapunov exponent for the classical Pauli rate equation for the probability vector corresponding to the eigenvalues of time evolved density operator $\bm{\rho}_t$. Then basic tools from Lyapunov theory allow one to derive the bound (\ref{CON}) (cf.~\cite{JPA2025} for details of the proof).

In this paper, we first present a new purely algebraic proof of (\ref{CON}) that does not rely on Lyapunov theory at all. Furthermore, the proof does not use the explicit form of the generator (\ref{LGKS}) but only uses the fact that $\mathfrak{L}$ defines a so-called conditionally completely positive map \cite{Evans1,Evans2}. {This, in turn, leads to the natural question of whether complete positivity is even needed for (\ref{CON}) to hold.} It should be stressed that the previously mentioned proof based on Lyapunov theory
explicitly relies on the assumption of positive rates $\gamma_\ell \geq 0$, otherwise only providing a possibly non-tight upper bound in terms of their absolute values.

Surprisingly, it turns out that complete positivity is not essential for the bound (\ref{CON}). We show that if $\mathfrak{L}$ generates a semigroup of 2-positive trace-preserving maps, then the corresponding relaxation rates satisfy (\ref{CON}). Recall that for a qubit system 2-positivity and complete positivity coincide and hence indeed both (\ref{LT}) and (\ref{Gen}) require complete positivity. For $d$-level systems with $d>2$, however, only 2-positivity matters. Actually, an interesting debate has developed over the question of whether complete positivity is necessary to represent quantum-dynamics \cite{PhysRevLett.73.1060,Shaji1,Shaji2,Terno,DoLi2016,Laflamme2023}.

{Finally, we relax the condition of 2-positivity and consider semigroups of unital Schwarz maps (in the Heisenberg picture). Let us recall that a unital map $\Phi$ is called a {\em Schwarz map} if it satisfies the operator Schwarz inequality, i.e.
\begin{equation}   \label{OSI}
    \Phi(X^\dagger X) \geq \Phi(X)^\dagger \Phi(X) ,
\end{equation}
for all operators $X$. Clearly any Schwarz map is necessarily positive. However, there exist positive unital maps which do not satisfy (\ref{OSI}), the simplest example being a transposition map.  Interestingly, any positive unital map satisfies (\ref{OSI}) for all Hermitian operators $X$ (in this case one calls (\ref{OSI}) a Kadison inequality \cite{Kadison1952,Kadison83}) \cite[Thm.~2.3.2]{Bhatia07}. It is well known \cite{Paulsen} that any 2-positive unital map satisfies (\ref{OSI}) and hence Schwarz maps interpolate between unital positive and 2-positive maps. In particular any unital completely positive map satisfies (\ref{OSI}). It is, therefore, natural to ask whether relaxing 2-positivity to the Schwarz property also induces a universal constraint for the relaxation rates. It turns out that (\ref{CON}) is modified to the following universal constraint
\begin{equation} \label{CON-S}
 \Gamma_{\rm max} \leq \frac{2}{d+1} \sum_{\ell=1}^{d^2-1} \Gamma_\ell \,,
\end{equation}
where $\frac2{d+1}$ is the harmonic mean of $\frac1d$ (prefactor in the $2$-positive case) and $1$ (prefactor in the positive case); this may be seen as a manifestation of the fact that Schwarz maps interpolate between positivity and $2$-positivity.
In the case of qubits the above constraint reduces to $\Gamma_{\rm max} \leq \frac{2}{3} \sum_{\ell=1}^{3} \Gamma_\ell $ which was derived recently in \cite{CKM24}. In summary, we establish the following tight constraints for relaxation rates }
\begin{tcolorbox}
\begin{equation}\label{mainres1}
 \Gamma_k \leq c_d \sum_{\ell=1}^{d^2-1}\Gamma_\ell \,,
\end{equation}
where the universal constant $c_d$ reads
\begin{equation}\label{mainres2}
 c_d =
 \begin{cases}
 \begin{array}{ll} 1\ , & \ \ \mbox{for positive trace-preserving dynamics}
 \\[0.4cm]
 \dfrac {2}{d+1}\ , & \ \ \mbox{for Schwarz dynamics}
 \\[0.4cm]
 \dfrac {1}{d} \ , & \ \ \mbox{for at least 2-positive trace-preserving dynamics} \end{array}
 \end{cases}
\end{equation}
\end{tcolorbox}

{Our results also admit an extension to  master equations governed by time-dependent generators
 $\dot{\bm{\rho}}_t = \mathfrak{L}_t(\bm{\rho}_t)$, where $\{\mathfrak{L}_t\}_{t \ge 0}$ is a family of sufficiently regular generators to satisfy global existence and uniqueness of solutions.
It gives rise to a dynamical map, i.e., a family of CPTP maps $\{\Lambda_{t,0}\}_{t\geq 0}$, defined by the following formula
\begin{equation}
    \Lambda_{t,0} = \mathcal{T} \exp\left( \int_0^t \mathfrak{L}_\tau d\tau \right) ,
\end{equation}
where $\mathcal{T}$ stands for a chronological time-ordering operator. Note that for any $s< t$ one finds
    $\Lambda_{t,0} = \Lambda_{t,s} \circ \Lambda_{s,0}$,
where we introduced a 2-parameter family of {\em propagators}
 $   \Lambda_{t,s}  := \Lambda_{t,0} \circ \Lambda_{s,0}^{-1} = \mathcal{T} \exp\left( \int_s^t \mathfrak{L}_\tau d\tau \right)$.
This family enjoys the following composition law (which immediately follows from the basic properties of the first-order linear differential equation)
\begin{equation}
    \Lambda_{t,s}  = \Lambda_{t,u} \circ \Lambda_{u,s} ,\ \ \ \ \   \forall \ t\ge u \ge s\ge 0 .
\end{equation}
However, it should be stressed that---contrary to $\Lambda_{t,0}$---the propagators $\Lambda_{t,s}$ need not be completely positive for $s > 0$.
In what follows we call the quantum evolution represented by $\{\Lambda_{t,0}\}_{t \geq 0}$ \textit{Markovian} if each propagator $\Lambda_{t,s}$ is CPTP for all $s \leq t$. The above property is often referred to as CP-divisibility \cite{RiHuPl2010,NM1} or time-dependent Markovianity \cite{Wolf08a,OSID_thermal_res}. In particular if $\mathfrak{L}_t\equiv\mathfrak{L}$ is time-independent, then $ \Lambda_{t,s} = e^{(t-s)\mathfrak{L}}$ is CPTP whenever $\mathfrak{L}$ is a legitimate GKLS generator. Hence, the concept of Markovianity defined via CP-divisibility of the corresponding dynamical map is a natural generalization of Markovian semigroups. It should be stressed, however, that in the literature there are other approaches to (non-)Markovianity \cite{NM1,NM2,NM3,ChrD2022,WoEiCuCi2008}. In particular the
review \cite{NM4} shows an intricate hierarchy of various notions and clearly indicates that quantum non-Markovianity is highly context-dependent.
An evolution is Markovian (CP-divisible) iff the corresponding generator $\mathfrak{L}_t$ has the standard GKLS form for all $t\geq 0$, i.e., it is represented by (\ref{LGKS}) with time-dependent Hamiltonian $\operatorname{H}_t$, noise operators $\operatorname{L}_{\ell,t}$, and transition rates $\gamma_{\ell,t} \geq 0$. Hence, whenever at least one of the rates  $\gamma_{\ell,t}$ is temporally negative the evolution is non-Markovian. Now, the relaxation rates $\Gamma_{\ell,t}$ are also time-dependent. In Section \ref{S-LT} it is shown that if all propagators $\Lambda_{t,s}$ are 2-positive for $t > s$, then the time-local relaxation rates $\Gamma_{\ell,t}$ satisfy the time-dependent version of (\ref{CON}). In particular this applies to any Markovian (CP-divisible) evolution.
}

The paper is organized as follows. In Section \ref{S-II} we provide an algebraic proof of (\ref{CON}) for semigroups of completely positive trace-preserving maps. The proof is split into four simple steps.  The key advantage of the presented proof is that it does not use any explicit form of the generator (\ref{LGKS}). In particular, the proof does not make direct use of the condition that all rates $\gamma_\ell$ are positive. This fact opens the possibility to generalize (\ref{CON}) beyond the scenario where the dynamics are completely positive. In Section \ref{S-III} it is shown that (\ref{CON}) holds for all generators giving rise to 2-positive trace-preserving semigroups and Section \ref{S-IV} discusses semigroups of unital Schwarz maps. {Section \ref{S-V} shows that the bounds \eqref{mainres1}–\eqref{mainres2} also imply the upper bound on the number of stationary states of such dynamics, as recently discussed in \cite{Facchi-Amato}.} {Time-dependent generators are discussed in Section \ref{S-LT}}.
Final conclusions are collected in Section   \ref{S-C}. Technical proofs of several propositions are included in the Supplementary Information File \cite{SIF}.

\section{Constraint for completely positive trace-preserving case}   \label{S-II}

{We start our analysis by proving the following
\begin{thm} \label{Pro-0} The relaxation rates $\Gamma_\ell$ of an arbitrary GKLS generator of a $d$-level quantum system satisfy
    \begin{equation} \label{CON-!}
 \Gamma_{\rm max} \leq \frac 1d \sum_{\ell=1}^{d^2-1} \Gamma_\ell \, .
\end{equation}
\end{thm}
The above theorem was originally proved in \cite{JPA2025} using basic tools from Lyapunov theory.}
Here we provide a purely algebraic proof which relies only on the defining property of the GKLS generator: $\mathfrak{L}$ generates a semigroup of completely positive maps iff it is conditionally completely positive \cite{Evans1,Evans2,Wolf08b,ChrD2022}, that is,
\begin{equation} \label{CCP}
 (\oper_d \otimes \oper_d - P^+)({\rm id}_d \otimes \mathfrak{L})(P^+) (\oper_d \otimes \oper_d - P^+) \geq 0\,,
\end{equation}
where $P^+ := |\psi^+\>\<\psi^+|$, $|\psi^+\> = \frac{1}{\sqrt{d}}\sum_{i=1}^d |i \otimes i\>$ is the canonical maximally entangled state in $\mathbb{C}^d \otimes \mathbb{C}^d$, and $\{|1\>,\ldots,|d\>\}$ is the standard orthonormal basis in $\mathbb{C}^d$. Note that (\ref{CCP}) is strictly weaker than the condition of complete positivity which requires that the Choi matrix of the map $\Phi$ is positive
\begin{equation} \label{C-Phi}
 C_\Phi := ({\rm id}_d \otimes \Phi)(P^+) \geq 0 \,.
\end{equation}
Indeed, conditional complete positivity requires positivity of the Choi matrix of $\mathfrak{L}$ \textit{only} on the subspace orthogonal to $|\psi^+\>$. Now, the spectrum of $\mathfrak{L}$ is symmetric w.r.t.~the real line (because $\mathfrak{L}$ is Hermiticity-preserving) and all eigenvalues $\lambda_\alpha $ of $\mathfrak{L}$ satisfy ${\rm Re}\lambda_\alpha \leq 0$.  Finally, $\mathfrak{L}$ admits a leading eigenvalue $\lambda_0 =0$ which corresponds to the stationary state $\mathfrak{L}(\bm{\omega})=0$. With this, the relaxation rates are defined as
\begin{equation}
 \Gamma_\ell = - {\rm Re}\, \lambda_\ell \ , \ \ \ \ell=1,\ldots,d^2-1 \,.
\end{equation}
It is {also} evident that
\begin{equation} \label{TrG}
 \sum_{\ell=1}^{d^2-1} \Gamma_\ell = -{\rm Tr}\mathfrak{L} \,.
\end{equation}
{Let us stress again that} the above properties of the spectrum of the generator have nothing to do with complete positivity{, and that any} generator giving rise to a semigroup of positive and trace-preserving maps enjoys the same properties.

 Now, we are ready to provide the proof {of \eqref{CON}} which, for the reader convenience, is divided into {four} simple steps.
\begin{itemize}
 \item Step 1. Given a linear map $\Phi$ let us denote by $\Phi^\dagger$ its adjoint w.r.t.~the Hilbert-Schmidt inner product
\begin{equation}
 (\Phi(A),B)_{\rm HS} = (A,\Phi^\dagger(B))_{\rm HS} \,,
\end{equation}
where $(A,B)_{\rm HS} := {\rm Tr}(A^\dagger B)$. In particular the adjoint of the GKLS generator (\ref{LGKS}) reads as follows
\begin{align}
	\mathfrak{L}^\dagger(\bm{{X}}) = \imath\,\left[\operatorname{H},\bm{X}\right]+
\sum_{\ell=1}^{d^{2}-1}\gamma_{\ell} \left( \operatorname{L}^\dagger_{\ell}\bm{X} \operatorname{L}_{\ell} - \frac 12 \{ \operatorname{L}_{\ell}^\dagger \operatorname{L}_{\ell},\bm{X}\} \right) \,,
		\label{LGKS-ad}
	\end{align}
and it satisfies $\mathfrak{L}^\dagger(\oper_d)=0$. Such a generator $\mathfrak{L}^\dagger$ gives rise to a semigroup of completely positive unital maps. It is worth noticing that the authors of \cite{GKLS} adopted Schr\"odinger's picture of the dynamics to derive (\ref{LGKS}), whereas Lindblad \cite{LinG1976} obtained (\ref{LGKS-ad}) using Heisenberg's picture. Both $\mathfrak{L}$ and $\mathfrak{L}^\dagger$ have the same spectrum (again, due to Hermiticity-preservation) and hence they have the same relaxation rates $\Gamma_\ell$. Therefore, the universal constraint (\ref{CON}) does not depend on whether one uses the Schr\"odinger or the Heisenberg picture. It is clear that $\mathfrak{L}^\dagger$ generates a semigroup of completely positive maps if and only if it is conditionally completely positive, that is, if it satisfies (\ref{CCP}).

Let us now assume that $\mathfrak{L}$ admits a faithful stationary state $\bm{\omega}$, in other words that there exists a strictly positive state operator $\bm{\omega}>0 $ satisfying $\mathfrak{L}(\bm{\omega})=0$. We explain in Remark~1 below why such an assumption is not restrictive. Next, we recall the definition of the \emph{Kubo-Martin-Schwinger} (KMS) inner product w.r.t.~$\omega$:
\begin{equation} \label{KMS}
 (A,B)_{\bm{\omega}} := (A,\bm{\omega}^{\frac 12}B\bm{\omega}^{\frac 12})_{\rm HS} .
\end{equation}
This is just a special case of the following family of inner products:
\begin{equation} \label{s}
 (A,B)_s := (A,\bm{\omega}^{s}B\bm{\omega}^{1-s})_{\rm HS} \ , \ \ \ s\in [0,1] .
\end{equation}
For $s=0$ this recovers the well-known \emph{Gelfand–Naimark–Segal} (GNS) inner product and for $s=1/2$ it reduces to the symmetric formula (\ref{KMS}) (cf.~\cite{Fagnola2010,AmCa2021} for more details). Denote by $\mathfrak{L}^\#$ the adjoint of $\mathfrak{L}^\dagger$ w.r.t. the KMS inner product
\begin{equation}\label{eq:def_L_sharp}
 (\mathfrak{L}^\#(A),B)_{\bm{\omega}} := (A,\mathfrak{L}^\dagger(B))_{\bm{\omega}} \,.
\end{equation}
Combining (\ref{KMS}) with (\ref{eq:def_L_sharp}) readily implies the following formula for $ \mathfrak{L}^\#$
\begin{equation}\label{eq:L_sharp_L_relation}
 \mathfrak{L}^\# = V^{-1}_{\bm{\omega}} \circ \mathfrak{L} \circ V_{\bm{\omega}} \,,
\end{equation}
where
\begin{equation}
 V_{\bm{\omega}}(X) = \bm{\omega}^{\frac 12}X\bm{\omega}^{\frac 12} \,, \ \ \ \
 V^{-1}_{\bm{\omega}}(X) = \bm{\omega}^{-\frac 12}X\bm{\omega}^{-\frac 12} \,.
\end{equation}
Note that if $\bm{\omega} = \oper_d/d$ is the maximally mixed state, then (\ref{eq:L_sharp_L_relation}) correctly reproduces $\mathfrak{L}^\# =(\mathfrak{L}^\dagger)^\dagger= \mathfrak{L}$.
\begin{proposition}
The KMS-adjoint $\mathfrak{L}^\#$ defines a  GKLS generator in the Heisenberg picture.
\end{proposition}
Indeed, $\mathfrak{L}^\#$ gives rise to the following semigroup
\begin{equation}
 e^{t \mathfrak{L}^\#} = e^{t V^{-1}_{\bm{\omega}} \circ \mathfrak{L} \circ V_{\bm{\omega}}} =
 V^{-1}_{\bm{\omega}} \circ e^{t \mathfrak{L}} \circ V_{\bm{\omega}} \,,
\end{equation}
which is evidently completely positive for all $t\geq 0$ as it is a composition of three completely positive maps: $ V_{\bm{\omega}}$, $ V^{-1}_{\bm{\omega}}$ and $e^{t \mathfrak{L}}$. Moreover, $e^{t \mathfrak{L}^\#} $ is unital, because $ \mathfrak{L}(\bm{\omega})=0$ implies
\begin{equation}
 \mathfrak{L}^\#(\oper_d) = V^{-1}_{\bm{\omega}} \circ \mathfrak{L}(\bm{\omega})=0 \,.
\end{equation}
Hence $ e^{t \mathfrak{L}^\#}$ is completely positive and unital for all $t \geq 0$.

 \item Step 2. Since both $\mathfrak{L}^\dagger$ and $\mathfrak{L}^\#$ are GKLS generators in the Heisenberg picture, it is clear that
\begin{equation} \label{LLL}
 \widetilde{\mathfrak{L}} := \frac 12 ( \mathfrak{L}^\dagger + \mathfrak{L}^\#) \,,
\end{equation}
defines a GKLS generator as well.
Moreover, by construction, $ \widetilde{\mathfrak{L}}$ is self-adjoint (w.r.t.~to the KMS inner product) and hence the spectrum of $
\widetilde{\mathfrak{L}}$ is purely real. Note that the generators $\mathfrak{L}^\dagger$ and $\mathfrak{L}^\#$ are isospectral and hence
\begin{equation}
 {\rm Tr}\, \widetilde{\mathfrak{L}} = {\rm Tr}\, \mathfrak{L}^\dagger =  {\rm Tr}\, \mathfrak{L} \, .
\end{equation}
Denote relaxation rates of $\widetilde{\mathfrak{L}}$ by $\widetilde{\Gamma}_\ell$.

\begin{proposition}  \label{Pro-tilde}  If the relaxation rates $\widetilde{\Gamma}_\ell$ of $\widetilde{\mathfrak{L}}$
satisfy (\ref{CON-!}), then relaxation rates ${\Gamma}_\ell$ of the original generator ${\mathfrak{L}}$ satisfy (\ref{CON-!}) as well.
\end{proposition}
Indeed, if
\begin{equation}
 \widetilde{\Gamma}_{\rm max} \leq \frac 1d {\rm Tr}(- \widetilde{\mathfrak{L}}) = \frac 1d {\rm Tr}(- \mathfrak{L})\, ,
\end{equation}
then to complete the proof it suffices to show that ${\Gamma}_{\rm max}$ of $\mathfrak{L}$ is upper bounded by $\widetilde{\Gamma}_{\rm max}$ of $\widetilde{\mathfrak{L}}$. But this follows immediately from the Bendixson--Hirsch inequality \cite{Bendixson02,Hirsch02} (cf.~also \cite[Thm.~12.6.6]{Mirsky55}): for any complex matrix $A\in \Md $ with (complex) eigenvalues $\{a_1,\ldots,a_d\}$
\begin{eqnarray}
 m \leq {\rm Re}\, a_k \leq M \,,
\end{eqnarray}
where $m$ and $M$ are minimal and maximal eigenvalues of the Hermitian matrix $\frac 12(A + A^\dagger)$. Applying this inequality to $\mathfrak{L}^\dagger$ and the self-adjoint $\widetilde{\mathfrak{L}}$ one finds
\begin{equation}
 -\widetilde{\Gamma}_{\rm max} \leq -{\Gamma}_{\rm max} \leq 0 \,,
\end{equation}
which altogether completes the proof of Proposition \ref{Pro-tilde}. To summarize, in this step we reduced the original problem for arbitrary generators to self-adjoint generators with purely real spectra.

\item Step 3. Given any orthonormal basis {$\{|e_1\rangle,\ldots,|e_d\rangle\}$} in $\mathbb{C}^d${, we define}
\begin{equation} \label{Kij}
 \mathcal{K}_{ij} := {\rm Tr}(P_i \mathfrak{L}(P_j)) \,,
\end{equation}
with $P_i := |e_i\>\<e_i|$.
{Note that, as a direct consequence of the above definition, $\mathcal{K}_{ij}$ satisfies the following properties:
\begin{equation}
 \mathcal{K}_{ij} \geq 0 \ \ \ (i \neq j) \ , \ \ \ \sum_{i=1}^d \mathcal{K}_{ij}=0;
\end{equation}
and hence defines a generator of a classical semigroup of stochastic matrices \cite{Kampen2007}.}
We stress that $\mathcal{K}_{ij}$ does depend on the basis $\{|e_i\rangle \}_{i=1}^d$.

\begin{proposition} \label{Pro-I}
Any quantum GKLS generator $\mathfrak{L}$, together with the corresponding classical generator $\mathcal{K}$ defined with respect to an arbitrary orthonormal basis $\{|e_i\rangle\}_{i=1}^d$, satisfies
\begin{equation} \label{KL}
  \operatorname{Tr} \mathfrak{L} \leq d\, \operatorname{Tr} \mathcal{K}.
\end{equation}
\end{proposition}
\begin{proof}
Recall that given a linear map $\Phi : \Md \to \Md$ one defines
\begin{equation}
 {\rm Tr}\,\Phi = \sum_{\alpha =1}^{d^2} (F_\alpha,\Phi(F_\alpha))_{\rm HS} = \sum_{\alpha =1}^{d^2} {\rm Tr}(F_\alpha^\dagger \Phi(F_\alpha)) \,,
\end{equation}
where $F_\alpha$ is an arbitrary orthonormal (i.e., $(F_\alpha,F_\beta)_{\rm HS} = \delta_{\alpha\beta}$) basis in $\Md$. In particular, taking $F_\alpha = |i\>\<j|$ one finds
\begin{equation}
 {\rm Tr}\,\Phi = \sum_{i,j =1}^{d} \<i|\, \Phi(|i\>\<j|)\, |j\>\,.
\end{equation}
Note that
\begin{equation} \label{tr-psi}
 {\rm Tr}\,\Phi = d^2 \< \psi^+| C_\Phi | \psi^+\> \,,
\end{equation}
where $C_\Phi$ denotes the Choi matrix \eqref{C-Phi} of $\Phi$ (See \cite[Lemma~2]{vE24_decomp_findim}).
Now, starting from the condition (\ref{CCP}) of conditional complete positivity of the generator $\mathfrak{L}$
\begin{equation} \label{}
 (\oper_d \otimes \oper_d - P^+) C_{\mathfrak{L}} (\oper_d \otimes \oper_d - P^+) \geq 0\,,
\end{equation}
in particular, one has
\begin{equation} \label{}
 \sum_{k=1}^d \<k \otimes k| (\oper_d \otimes \oper_d - P^+) C_{\mathfrak{L}} (\oper_d \otimes \oper_d - P^+) |k \otimes k\> \geq 0\,.
\end{equation}
Using $\< k \otimes k|\psi^+\> = \frac{1}{\sqrt{d}}$ one easily computes
\begin{eqnarray*}
&& \sum_{k=1}^d \< k \otimes k| C_{\mathfrak{L}}|k \otimes k\> = \frac 1d\sum_{k=1}^d \<k|\, \mathfrak{L}(|k\>\<k|)\, |k\> = \frac 1d {\rm Tr}\, \mathcal{K} \,, \\
&& \sum_{k=1}^d \< k \otimes k| P^+_d\, C_{\mathfrak{L}}|k \otimes k\> = \<\psi^+|\, C_{\mathfrak{L}}\, | \psi^+\> = \frac{1}{d^2} {\rm Tr}\, \mathfrak{L} \,, \\
&& \sum_{k=1}^d \< k \otimes k| C_{\mathfrak{L}}\,P^+_d |k \otimes k\> = \<\psi^+|\, C_{\mathfrak{L}}\, | \psi^+\> = \frac{1}{d^2} {\rm Tr}\, \mathfrak{L} \,, \\
&& \sum_{k=1}^d \< k \otimes k| P^+_d \, C_{\mathfrak{L}} \,P^+_d|k \otimes k\> = \<\psi^+|\, C_{\mathfrak{L}}\, | \psi^+\> = \frac{1}{d^2} {\rm Tr}\, \mathfrak{L} \,,
\end{eqnarray*}
and hence
\begin{equation}
 {\rm Tr}\, \mathcal{K} - \frac 1d {\rm Tr}\, \mathfrak{L} \geq 0 \,,
\end{equation}
which immediately proves (\ref{KL}).
\end{proof}

It should be  emphasized that {the proof of (\ref{KL}) does not rely on the explicit representation (\ref{LGKS}). In particular, it does not require the positivity of the transition rates $\gamma_\ell \geq 0$ and it is entirely based on the very notion of conditional complete positivity of $\mathfrak{L}$.}  {Furthermore, as} we show later, (\ref{KL}) can be proved for a wider class of generators assuming that they generate a semigroup of 2-positive maps. Observe that, (\ref{KL}) together with (\ref{TrG}) implies
\begin{equation}
 {\rm Tr}(- \mathcal{K}) \leq \frac 1d \sum_{\ell=1}^{d^2-1} \Gamma_\ell \,.
\end{equation}
Hence, if we manage to find {(as we do in the next step)} a particular basis $\{|e_i\rangle \}_{i=1}^d$ for which the classical generator $\mathcal{K}$ has an eigenvalue whose real part coincides with `$-\Gamma_{\rm max}$' then, necessarily, $\Gamma_{\rm max} \leq {\rm Tr}(- \mathcal{K})$, because (as in the quantum case) all eigenvalues of `$-\mathcal{K}$' have non-negative real parts and hence (\ref{CON}) immediately follows.

\item Step 4. Recall that Step 2. reduced the problem to generators with purely real spectra. If $\lambda$ is a real eigenvalue of $\mathfrak{L}$, then the corresponding eigenvector $X$ can always be chosen to be Hermitian. Indeed, given some non-Hermitian $Y\neq 0$ such that $\mathfrak{L}(Y) = \lambda Y $, one has $\mathfrak{L}(Y^\dagger) = \lambda Y^\dagger $ (because $\mathfrak{L}$ is Hermiticity-preserving) and hence
$X:=i(Y-Y^\dagger)\neq 0$ is a Hermitian eigenvector corresponding to $\lambda$. Let $\{|k\>\}_{k=1}^d$ be an eigenbasis of $X$, i.e.,
$X = \sum_k x_k |k\>\<k|$ with $x_k \in \mathbb{R}$. Define a classical generator $\mathcal{K}$ w.r.t.  an eigenbasis of $X$, i.e., $\mathcal{K}_{ij} = \<i|\mathfrak{L}(|j\>\<j|)|i\>$.
\begin{proposition}    \label{PRO-2}
A real eigenvalue  $\lambda$ of $\mathfrak{L}$ is an eigenvalue of $\mathcal{K}$ and $\bm{x} = (x_1,\ldots,x_d)$ defines the corresponding eigenvector:
\begin{equation}
 \mathcal{K} \bm{x} = \lambda \bm{x} \,.
\end{equation}
\end{proposition}
Indeed, one immediately finds
\begin{eqnarray}\label{eq:K_Gamma_max}
 (\mathcal{K} \bm{x})_i &=& \sum_{j=1}^d \mathcal{K}_{ij} x_j = \sum_{j=1}^d \<i|\mathfrak{L}(x_j|j\>\<j|)|i\> = \<i|\mathfrak{L}(X)|i\> = \lambda \<i|X|i\> =  \lambda x_i.
\end{eqnarray}
Therefore, assuming that ${\mathfrak{L}}$ has a real spectrum one finds that `$-\Gamma_k$' is an eigenvalue of $\mathcal{K}$ which, in turn, implies that $\Gamma_k \leq {\rm Tr}(- \mathcal{K})$. Finally, taking into account (\ref{KL}) one finds
\begin{equation}   \label{<<}
    \Gamma_k \leq {\rm Tr}(- \mathcal{K}) \leq \frac 1 d {\rm Tr}(- \mathfrak{L}) \equiv \frac 1d \sum_{\ell=1}^{d^2-1} \Gamma_\ell ,
\end{equation}
which ends the proof of Theorem \ref{Pro-0}.   \hfill $\Box$
\end{itemize}

\begin{remark}\label{remark_posdef_stationary_state_dense}
In Step 1 of our proof, we assumed that $\bm{\omega}$ is a faithful stationary state of $\mathfrak{L}$. While not every generator $\mathfrak{L}$ satisfies this, \textnormal{almost all}  do, in the sense that the set of all GKLS generators which have a full-rank stationary state $\bm{\omega}$ is dense in the space of all GKLS generators (for a proof see \cite{SIF}).
Hence, there exists a sequence $\{ \mathfrak{L}_n\}_{n \in \mathbb{N}}$ of generators which all have a full-rank stationary state $\bm{\omega}_n$ such that
\begin{equation}
 \mathfrak{L} = \lim_{n \to \infty} \mathfrak{L}_n\,.
\end{equation}
Together with the known fact that the eigenvalues, resp.~relaxation rates depend continuously on the input \cite[Ch.~II, Thm.~5.14]{Kato80}, this generalizes the inequality from generators with full-rank stationary state to arbitrary generators.
\end{remark}
\begin{remark}
The reader might wonder why in Step 1 we decided to introduce the KMS inner product instead of the standard Hilbert-Schmidt one. Clearly, we could have instead defined
\begin{equation}
 \widetilde{\mathfrak{L}} = \frac 12 (\mathfrak{L}^\dagger  + \mathfrak{L}) \,,
\end{equation}
which is self-adjoint w.r.t.~the Hilbert-Schmidt inner product. Note, however, that contrary to (\ref{LLL}) the above formula does not define a proper generator (in the Heisenberg picture), i.e., $\widetilde{\mathfrak{L}}(\oper_d) \neq 0$ (unless $\bm{\omega}$ is maximally mixed). In particular, Step 4 need not apply to $\widetilde{\mathfrak{L}}$ because
now the minimal eigenvalue of the
associated $-\mathcal K$ may be negative, so the inequality $\Gamma_{\rm max} \leq {\rm Tr} (- \mathcal{K})$ used in (\ref{<<}) need not hold anymore.
\end{remark}

\begin{remark}
Recall that $$\mathfrak{L}(\bm{\rho}) = -\imath [H,\bm{\rho}]+ \mathfrak{D}(\bm{\rho})$$ enjoys a quantum detailed balance property w.r.t.~a steady state $\bm{\omega}$ if $[H,\omega]=0$ and the dissipative part of the generator $\mathfrak{D}$ satisfies \cite{Alicki1976,KFGV1977}
\begin{equation} \label{DD}
 {\rm Tr}(\bm{\omega} X^\dagger \mathfrak{D}^\dagger(Y)) = {\rm Tr}(\bm{\omega} \mathfrak{D}^\dagger(X^\dagger) Y) \,,
\end{equation}
i.e., $\mathfrak{D}^\dagger$ is self-adjoint w.r.t.~the GNS inner product $(X,Y)_{\rm GNS} = {\rm Tr}(\bm{\omega} X^\dagger Y)$. The authors of \cite{Fagnola2010} considered the so-called $s$-detailed balance, i.e., when $\mathfrak{D}^\dagger$ is self-adjoint w.r.t.~(\ref{s}). In particular, they analyzed the KMS-detailed balance  corresponding to a symmetric scenario $s=1/2$. Finally, the $s$-inner product and, in particular, the $s=1/2$ KMS case are instrumental to obtain quantum log-Sobolev inequalities \cite{KaTe2013,KaEi2013}, that provide a framework for the derivation of improved bounds on the convergence time of quantum-dynamical semigroups. {Interestingly, the family of $s$-detailed balance conditions  collapses only to two instances: the symmetric case $s=1/2$ and the strictly stronger case $s\neq 1/2$ \cite{Carlen-Maas}.}
\end{remark}

\section{Beyond completely positive semigroups: 2-positive maps} \label{S-III}

In this section, we show that the assumption of complete positivity may be significantly relaxed.
Let us recall that a positive map $\Phi : \Md \to \Md$ is called $k$-positive \cite{Stormer,Paulsen,Bhatia07} if the extended map $\Phi_k : \mathcal{M}_k(\mathbb{C})\otimes \Md \to \mathcal{M}_k(\mathbb{C})\otimes \Md$ defined by
\begin{equation}
 \Phi_k := {\rm id}_k \otimes \Phi \,,
\end{equation}
is positive (here ${\rm id}_k : \mathcal{M}_k(\mathbb{C}) \to \mathcal{M}_k(\mathbb{C})$ is the identity map).
In this language, complete positivity coincides with $d$-positivity \cite{Choi-72,Choi-75}. Anyway, it is clear that $k$-positivity requires
\begin{equation}
 \< \psi |\, [{\rm id}_k \otimes \Phi](|\phi\>\<\phi|)\, |\psi\> \geq 0 \,,
\end{equation}
for all vectors $\psi,\phi \in \mathbb{C}^k \otimes \mathbb{C}^d$. {Table \ref{Table1} summarizes the hierarchy of maps with different {\em degree}  of positivity:

\begin{table}[h!]
\centering
\renewcommand{\arraystretch}{1.2}
\begin{tabular}{|l|p{6cm}|p{5cm}|}
\hline
\textbf{Class of map} & \textbf{Definition} & \textbf{Properties} \\ \hline
Positive & $\Phi(X) \ge 0$ whenever $X \ge 0$ & Preserves positivity for single operators \\ \hline
2-Positive & ${\rm id}_2 \otimes \Phi$ is positive & Preserves positivity for $2\times 2$ block matrices \\ \hline
$k$-Positive & ${\rm id}_k \otimes \Phi $ is positive & Preserves positivity for $k\times k$ block matrices; stronger than $\ell$-positive for $\ell < k$ \\ \hline
Completely Positive (CP) & $\Phi$ is $k$-positive for all $k=1,2,\ldots$. It saturates for $k=d$. & Choi matrix $({\rm id}_d \otimes \Phi)(P^+)$ is positive semi-definite
\\ \hline
\end{tabular}
\caption{Hierarchy and properties of different classes of maps.   \label{Table1}}
\end{table}
}

Consider now a semigroup $e^{t \mathfrak{L}}$ of $k$-positive maps. One has the following result \cite{GoKoSu1976,Evans1}
\begin{proposition}
$\mathfrak{L}$ generates a semigroup of $k$-positive maps if and only if $\mathfrak{L}$ is conditionally $k$-positive, that is,
\begin{equation} \label{kCP}
 (\oper_k \otimes \oper_d - |\phi\>\<\phi|) [{\rm id}_k \otimes \mathfrak{L}](|\phi\>\<\phi|) (\oper_k \otimes \oper_d - |\phi\>\<\phi|) \geq 0 \,,
\end{equation}
for all normalized vectors $|\phi\rangle \in \mathbb{C}^k \otimes \mathbb{C}^d$.
\end{proposition}
Equivalently, condition (\ref{kCP}) may be reformulated as follows
\begin{equation}\label{eq:kCP_alt_form}
 \< \psi |\, [{\rm id}_k \otimes \mathfrak{L}](|\phi\>\<\phi|)\, |\psi\> \geq 0 \,,
\end{equation}
for all mutually orthogonal vectors $|\psi\rangle \perp |\phi \rangle$ from $\mathbb{C}^k \otimes \mathbb{C}^d$ (cf.~also \cite{Koss72}).
The Proposition \ref{Pro-I} admits the following generalization:
\begin{proposition} \label{Pro-II}
If $\mathfrak{L}$ generates a semigroup of 2-positive maps, then
\begin{equation} \label{KL2}
 {\rm Tr}\mathfrak{L}\leq d\,{\rm Tr}\, \mathcal{K} \,,
\end{equation}
for any orthonormal basis $\{|e_i\rangle \}_{i=1}^d$.
\end{proposition}
Again, the proof of (\ref{KL2}) is entirely based on the very notion of conditional 2-positivity of $\mathfrak{L}$. We defer the proof to \cite{SIF}.
Using (\ref{KL2}) we are in a position to prove:
\begin{corollary}
If $\mathfrak{L}$ generates a semigroup of 2-positive trace-preserving maps, then the corresponding relaxation rates satisfy (\ref{CON}).
\end{corollary}
Indeed, let us check that all steps from the proof of the completely positive case (Section~\ref{S-II}) work here, as well. Step 1: let $\bm{\omega}$ be a positive definite stationary state of $\mathfrak{L}$ (recall from Remark~\ref{remark_posdef_stationary_state_dense} that this is not a restriction, also because the density result in \cite{SIF} was formulated for general $k$-positivity).
As the adjoint $\mathfrak{L}^\dagger$ defines a generator of a 2-positive unital dynamical semigroup, the KMS adjoint $\mathfrak{L}^\#$ of $\mathfrak{L}^\dagger$ satisfies
\begin{equation}
 \mathfrak{L}^\# = V^{-1}_{\bm{\omega}} \circ \mathfrak{L} \circ V_{\bm{\omega}} \,,
\end{equation}
 thus defining the generator of a 2-positive unital dynamical semigroup in the Heisenberg picture. Indeed, $\mathfrak{L}^\#$ gives rise to
\begin{equation}
 e^{t \mathfrak{L}^\#} = e^{t V^{-1}_{\bm{\omega}} \circ \mathfrak{L} \circ V_{\bm{\omega}}} =
 V^{-1}_{\bm{\omega}} \circ e^{t \mathfrak{L}} \circ V_{\bm{\omega}} \,,
\end{equation}
which is evidently 2-positive being a composition of two completely positive maps: $ V_{\bm{\omega}}$, $ V^{-1}_{\bm{\omega}}$ and the 2-positive map $e^{t \mathfrak{L}}$. Step 2 also works without any change due to the fact that the self-adjoint generator
\begin{equation}
 \widetilde{\mathfrak{L}} := \frac 12 ( \mathfrak{L}^\dagger + \mathfrak{L}^\#) \,,
\end{equation}
defines a generator of 2-positive trace-preserving maps, so Step 2 applies to it and one can apply the Bendixson--Hirsch inequality.  Finally, Step 4 does not depend on the assumption of complete positivity (it is true for any Hermiticity-preserving map).
\begin{remark}
It should be stressed that the assumption of 2-positivity is absolutely essential and cannot be reduced to simple positivity. {To illustrate this,} consider the following qubit generator
\begin{equation} \label{Pauli}
 \mathfrak{L}(\bm{\rho}) = \frac 12 \sum_{k=1}^{{3}} \gamma_k (\sigma_k \bm{\rho} \sigma_k - \bm{\rho})\,.
\end{equation}
{One finds for the spectrum: $\mathcal{L}(\oper_2) = 0$, together with
\begin{equation}\label{}
  \mathcal{L}(\sigma_k) = \lambda_k \sigma_k ,
\end{equation}
with $\lambda_1 = -(\gamma_2+{\gamma_3})$, $\lambda_2 = -(\gamma_3+\gamma_1)$, and $\lambda_3 = -(\gamma_1+\gamma_2)$. Hence, the corresponding relaxation rates read
\begin{equation}\label{}
  \Gamma_1 = \gamma_2 + {\gamma_3} \ , \ \ \Gamma_2 = \gamma_3+\gamma_1 \ ,\ \ \Gamma_3 = \gamma_1+\gamma_2 .
\end{equation}
}
 On the other hand, $\mathfrak{L}$ generates a semigroup of positive maps if and only if the transition rates satisfy \cite{ChrD2022}
$$
\gamma_1 + \gamma_2 \geq 0 \ , \ \ \gamma_2 + \gamma_3 \geq 0 \ , \ \ \gamma_3 + \gamma_1 \geq 0 \,,
$$
which means that at most one rate can be negative. Consider $\gamma_1=\gamma_2=1$ and $\gamma_3=-1$. One finds for the relaxation rates
$$
\Gamma_1 = \Gamma_2 = 0 \ , \ \ \Gamma_3 = 2 \,,
$$
and hence $\Gamma_{\rm max}= \sum_k \Gamma_k$ so the constraint (\ref{CON}) is violated.
\end{remark}
This example proves:
\begin{corollary}
If $\mathfrak{L}$ generates a semigroup of positive trace-preserving maps, then in general the relaxation rates of $\mathfrak{L}$ satisfy only the trivial {(yet tight)} constraint
\begin{equation} \label{trivial}
 \Gamma_{\rm max} \leq \sum_{\ell=1}^{d^2-1} \Gamma_\ell\,.
\end{equation}
\end{corollary}
This shows that there is {a} critical difference between positivity and 2-positivity.

\section{Universal constraint for semigroups of Schwarz maps} \label{S-IV}

As seen just now we cannot replace 2-positivity by 1-positivity while still getting a non-trivial constraint on the relaxation rates; for the 1-positive scenario we saw that one cannot go beyond the trivial bound (\ref{trivial}). Interestingly, there exists a class of maps which interpolate between positive and 2-positive maps: the so-called Schwarz maps \cite{Paulsen}. A unital linear map $\Phi : \Md \to \Md$ is called a Schwarz map if
\begin{equation}\label{KS}
 \Phi(X^\dagger X) \geq \Phi(X)^\dagger \Phi(X) \,,
\end{equation}
for all $X \in \Md$ \cite{Paulsen,Kadison83,Bhatia07,Stormer}. All Schwarz maps are necessarily positive.
A key example of a positive unital map which is not Schwarz is the transposition map $T$. Indeed, for $X = |1\>\<2|$ one finds $T(X^\dagger X) = T(|2\>\<2|) = |2\>\<2|$, whereas $T(X)^\dagger T(X) =|1\>\<1| $.

{Interestingly, any unital 2-positive map $\Phi : \Md \to \Md$ satisfies (\ref{KS}). Indeed, consider a positive block operator
$$   \mathbb{X} = \left( \begin{array}{cc} X^\dagger X & X^\dagger \\ X & \oper_d \end{array} \right) \ge 0 , $$
with arbitrary $X \in \Md$. Then if $\Phi$ is unital and 2-positive one has
$$  ({\rm id}_2 \otimes \Phi)(\mathbb{X}) = \left( \begin{array}{cc} \Phi(X^\dagger X) & \Phi(X^\dagger) \\ \Phi(X) & \oper_d \end{array} \right)  \geq 0 , $$
which is equivalent to (\ref{KS}). However, there are Schwarz maps which are not 2-positive. A simple example of such a map $\Phi : \mathcal{M}_2(\mathbb{C}) \to \mathcal{M}_2(\mathbb{C})  $ reads
$$  \Phi(X) = \frac 12 \left( \frac 12 \oper_2 {\rm Tr}\, X + X^{\rm T} \right) .  $$
In a recent paper \cite{Bihalan} the authors presented a large family of qubit maps and identified which of them are positive but not Schwarz, and which are Schwarz but not 2-positive (which for one qubit is just completely positivity).
The following proposition illustrates an interesting property of Schwarz maps.
\begin{proposition}[Variance contractivity]
Let $\Phi$ be a positive trace-preserving map with an invariant state $\bm{\omega}$. If $\Phi^\dagger$ is a Schwarz map, then for any operator $A \in \Md$  one has
\begin{equation}  \label{Var}
    \mathrm{Var}_\omega(\Phi^\dagger(A)) \;\le\; \mathrm{Var}_\omega(A),
\end{equation}
where $\mathrm{Var}_\omega(A) := \operatorname{Tr}(\bm{\omega} A^\dagger A) - |\operatorname{Tr}(\bm{\omega} A)|^2$ is the quantum variance of $A$ in a state $\bm{\omega}$.
\end{proposition}
\begin{proof} Indeed, one computes
\[  \mathrm{Var}_\omega(\Phi^\dagger(A)) = {\rm Tr}(\bm{\omega} \, \Phi^\dagger(A^\dagger)\Phi^\dagger(A)) - |\operatorname{Tr}(\bm{\omega} \Phi^\dagger(A))|^2 =
{\rm Tr}(\bm{\omega} \, \Phi^\dagger(A^\dagger)\Phi^\dagger(A)) - |\operatorname{Tr}(\bm{\omega}A)|^2 ,
\]
due to $\Phi(\bm{\omega}) = \bm{\omega}$. Now, using the Schwarz inequality one obtains
\[
 \mathrm{Var}_\omega(\Phi^\dagger(A)) \leq {\rm Tr}(\bm{\omega} \, \Phi^\dagger(A^\dagger A)) - |\operatorname{Tr}(\bm{\omega}A)|^2 = {\rm Tr}(\bm{\omega} \, A^\dagger A) - |\operatorname{Tr}(\bm{\omega}A)|^2 = \mathrm{Var}_\omega(A) ,
 \]
which completes the proof.
\end{proof}}
{
\begin{remark}
As can be seen by using Kadison’s inequality, (\ref{Var}) also holds for unital positive maps when restricted to Hermitian matrices.
The above proposition means that, in the presence of the Schwarz property, this can be extended to non-Hermitian operators as well.
In fact, the variance of non-Hermitian operators also has physical applications and can, for example, be measured through weak measurements \cite{Pati}.
\end{remark}
}
{\begin{remark} In analogy to $k$-positivity one says that $\Phi$ is a $k$-Schwarz map if ${\rm id}_k \otimes \Phi$ is a Schwarz map. Finally, $\Phi$ is completely Schwarz if it is $k$-Schwarz for all $k=1,2,\ldots$. Interestingly, $k$-Schwarz maps interpolate between unital $k$-positive and unital $(k+1)$-positive maps and completely Schwarz maps coincide with unital completely positive maps \cite{Paulsen}.
\end{remark}}

Since we consider unital maps, we pass to the Heisenberg picture and consider the corresponding Heisenberg picture generator $\mathfrak{L}^\dagger$. In his seminal paper \cite{LinG1976} Lindblad  derived the following result:
\begin{proposition}
\label{PRO-L} $\mathfrak{L}^\dagger$ generates a semigroup of unital Schwarz maps if and only if $\mathfrak{L}^\dagger(\oper_d)=0$ and
\begin{equation}\label{L!}
 \mathfrak{L}^\dagger(X^\dagger X) \geq \mathfrak{L}^\dagger(X^\dagger) X + X^\dagger \mathfrak{L}^\dagger(X) \,,
\end{equation}
for all $X \in \Md$.
\end{proposition}
Lindblad called (\ref{L!}) {\em dissipativity condition} and $\mathfrak{L}$ a {\em dissipative generator} \cite{LinG1976}. In a recent paper \cite{CKM24} the following constraint for relaxation rates of a qubit dissipative generator was derived.
\begin{proposition}
If $\mathfrak{L}$ is a qubit dissipative generator, i.e., $e^{t \mathfrak{L}^\dagger}$ is a unital Schwarz map for $t \geq 0$, then the corresponding relaxation rates satisfy
\begin{equation} \label{2/3}
 \Gamma_{\rm max} \leq \frac{2}{3} \sum_{\ell=1}^3 \Gamma_\ell\,.
\end{equation}
\end{proposition}
The purpose of this section  is to generalize (\ref{2/3}) to Schwarz semigroups on $\Md$ with {arbitrary} $d\geq 2$. Our generalization is based on the following relation between such $\mathfrak{L}$ and the corresponding classical generator $\mathcal{K}$ {\eqref{Kij}}:
\begin{proposition} \label{Pro-III}
If $\mathfrak{L}^\dagger$ generates a semigroup of unital Schwarz maps, then
\begin{equation} \label{KL3}
 {\rm Tr}\mathfrak{L}^\dagger = {\rm Tr}\mathfrak{L} \leq \frac{d+1}{2}\,{\rm Tr}\, \mathcal{K} \,,
\end{equation}
for any orthonormal basis $\{|e_i\rangle \}_{i=1}^d$.
\end{proposition}
The proof of (\ref{KL3}) relies on the dissipativity condition of $\mathfrak{L}$ and it is presented in \cite{SIF}.

\begin{corollary}
If $\mathfrak{L}^\dagger$ generates a semigroup of unital Schwarz maps, then the corresponding relaxation rates satisfy
\begin{equation} \label{Sd}
 \Gamma_{\rm max} \leq \frac{2}{d+1} \sum_{\ell=1}^{d^2-1} \Gamma_\ell\,.
\end{equation}
\end{corollary}
Clearly, (\ref{Sd}) reduces to (\ref{2/3}) for $d=2$ so it is indeed a generalization of the latter. The proof goes along the same line as for complete positivity and 2-positivity. Let $\bm{\omega}$ be a faithful stationary state of $\mathfrak{L}$. Again, this is not a restriction; the argument from Remark~\ref{remark_posdef_stationary_state_dense} still applies because the density result from \cite{SIF} continues to hold when replacing $k$-positivity, resp.~conditional $k$-positivity by Schwarz, resp.~dissipativity. The only point where properties of $e^{t\mathfrak{L}}$ (other than Hermiticity-preservation) is used is the existence of a stationary state in Step 1. But this is still true because every unital Schwarz map is positive and unital, meaning it has a fixed point because its adjoint (positive and trace-preserving) does, by the Brouwer fixed point theorem.
The KMS adjoint $\mathfrak{L}^\#$ of $\mathfrak{L}^\dagger$ satisfies
\begin{equation}
 \mathfrak{L}^\# = V^{-1}_{\bm{\omega}} \circ \mathfrak{L} \circ V_{\bm{\omega}} \,,
\end{equation}
thus defining a legitimate generator of unital Schwarz maps.
Indeed, $\mathfrak{L}^\#$ gives rise to the following semigroup
\begin{equation}
 e^{t \mathfrak{L}^\#} = e^{t V^{-1}_{\bm{\omega}} \circ \mathfrak{L} \circ V_{\bm{\omega}}} =
 V^{-1}_{\bm{\omega}} \circ e^{t \mathfrak{L}} \circ V_{\bm{\omega}} \,,
\end{equation}
which is evidently unital and Schwarz being a composition of completely positive maps: $ V_{\bm{\omega}}$, $ V^{-1}_{\bm{\omega}}$ and the unital Schwarz map $e^{t \mathfrak{L}^\dagger}$. Finally, Step 4 perfectly works without any change due to the fact that self-adjoint generator
\begin{equation}
 \widetilde{\mathfrak{L}} := \frac 12 ( \mathfrak{L}^\dagger + \mathfrak{L}^\#) \,,
\end{equation}
defines a generator of unital Schwarz maps and one can apply the Bendixson--Hirsch inequality.
\begin{remark}
Consider once again the qubit generator of Pauli maps defined in (\ref{Pauli}) which, notably, is self-adjoint ($\mathfrak{L}^\dagger=\mathfrak{L}$). Again to generate a semigroup of Schwarz maps it may have at most one negative eigenvalue. Assuming $\gamma_3<0$ it is well known that $\mathfrak{L}$ generates a semigroup of Schwarz maps if \cite{ChrD2022,CKM24}
$$
\gamma_1+ 2 \gamma_3 \geq 0 \ , \ \ \ \gamma_2 + 2 \gamma_3 \geq 0 \,.
$$
In particular for $\gamma_1=\gamma_2=2$ and $\gamma_3=-1$ one finds
$$
\Gamma_1 = \Gamma_2 = 1 \ , \ \ \Gamma_3 = 4 \,,
$$
and hence $\Gamma_{\rm max} = \frac 23 \sum_k \Gamma_k$. Thus the constraint (\ref{CON}) is violated and (\ref{Sd}) is saturated.
\end{remark}

\section{A bound for a number of steady states} \label{S-V}

In a recent paper \cite{Facchi-Amato} an interesting bound for the number $m_0$ of steady states of arbitrary non-trivial GKLS generators (equivalently $m_0$ is a dimension of the kernel of $\mathfrak{L}$) was found:
\begin{equation} \label{m0}
 m_0^{\rm (CP)} \leq d^2-2d +2\,,
\end{equation}
and this bound was even proven to be sharp. It turns out that the constraints on the relaxation rates of 2-positive trace-preserving dynamics, as well as of unital Schwarz semigroups imply the following bounds on steady states for these relaxed scenarios:
\begin{proposition}
If $\mathfrak{L}\neq 0$ generates a 2-positive trace-preserving semigroup, then
\begin{equation} \label{mo2P}
 m_0^{\rm (2P)} \leq d^2 -d\,.
\end{equation}
If $\mathfrak{L}^\dagger\neq 0$ generates a semigroup of unital Schwarz maps, then
\begin{equation} \label{moS}
 m_0^{\rm (S)} \leq d^2 -\frac{d+1}{2}\,.
\end{equation}
\end{proposition}
\noindent Proof. If $\mathfrak{L}$ generates 2-positive trace-preserving semigroup, then
\begin{eqnarray} \label{a}
 \Gamma_k \leq \frac 1d \sum_\ell \Gamma_\ell\,.
\end{eqnarray}
Now, following arguments from \cite{Facchi-Amato} one observes that summing (\ref{a}) over strictly positive rates one obtains
\begin{eqnarray} \label{a1}
 \sum_{\Gamma_k >0} \Gamma_k \leq \frac 1d \sum_{\Gamma_k >0} \Big( \sum_\ell \Gamma_\ell \Big)\,.
\end{eqnarray}
But $\sum_{\Gamma_k >0} = d^2 - m_0^{\rm (2P)}$ by definition of $m_0^{\rm (2P)}$ so re-arranging (\ref{a1}) and dividing by $\sum_{\Gamma_k >0} \Gamma_k$ ($\neq 0$, because $\mathfrak{L}\neq 0$) yields (\ref{mo2P}). Moreover, (\ref{moS}) is shown analogously. Clearly one has the following hierarchy of inequalities
\begin{equation}
 m_0^{\rm (CP)} \leq m_0^{\rm (2P)} \leq m_0^{\rm (S)}\,.
\end{equation}
For $d=2$ complete positivity and 2-positive coincide and hence $m_0^{\rm (CP)} = m_0^{\rm (2P)} =2$. For $d=3$
$$
m_0^{\rm (CP)} =5 < m_0^{\rm (2P)} =6 < m_0^{\rm (S)} =7 \,,
$$
i.e., all three bounds are different. It would be interesting to construct examples which saturate the bounds.

{
\section{Constraints for time-local relaxation rates}   \label{S-LT}
It has long been known \cite{ShiT1970} that open system evolution equations derived by means of the Nakajima-Zwanzig projection method can be couched, at least formally and over finite time intervals, into the form of time-local equations. In the case of open quantum systems, the result of this procedure is in general a time-dependent master equation
\begin{align}
\dot{\bm{\rho}}_t = \mathfrak{L}_{t}(\bm{\rho}_t) .
\label{me-t}
\end{align}
The generator $\mathfrak{L}_{t}$ is always amenable to the canonical form (see, e.g., \cite{HaCrLiAn2014})
\begin{align}
	\mathfrak{L}_t(\bm{{\rho}}) =-\imath\,\left[\operatorname{H}_t,\bm{\rho}\right]+
\sum_{\ell=1}^{d^{2}-1}\gamma_{\ell,t} \left( \operatorname{L}_{\ell,t}\bm{\rho} \operatorname{L}_{\ell,t}^\dagger - \frac 12 \{ \operatorname{L}_{\ell,t}^\dagger \operatorname{L}_{\ell,t},\bm{\rho}\} \right)
		\label{LGKS-t}
	\end{align}
with time-dependent Hermitian $\operatorname{H}_t$, noise operators $\operatorname{L}_{\ell,t}$ coupled by a collection of also time-dependent scalar function $\gamma_{\ell,t}$
for $\ell=1,\dots,d^2-1$.  However, contrary to the time-independent case leading to a Markovian semigroup, the rates $\gamma_{\ell,t}$ need not be positive for all $t \ge 0$.
Assuming the time-dependence of all these quantities to be sufficiently regular for any $t\ge 0$ to satisfy global existence and uniqueness of solutions, solving (\ref{me-t}) with the identity map as initial condition at $t= 0$ yields the following formula for the dynamical map
\begin{equation}
    \Lambda_{t,0} = \mathcal{T} \exp\left( \int_0^t \mathfrak{L}_\tau d\tau \right) = {\rm id}_d + \int_0^t d\tau\, \mathfrak{L}_\tau + \int_0^t dt_1 \int_0^{t_1} dt_2 \, \mathfrak{L}_{t_1} \circ \mathfrak{L}_{t_2} + \ldots .
\end{equation}
A time-dependent generator is physically admissible if $\Lambda_{t,0}$ is CPTP for all $t \geq 0$.  It should be stressed that in general there is no guarantee that the propagators
\begin{equation}   \label{Texp}
    \Lambda_{t,s}  = \mathcal{T} \exp\left( \int_s^t \mathfrak{L}_\tau d\tau \right) ,
\end{equation}
are also CPTP for $t \geq s > 0$. Existence and uniqueness of solutions also imply that propagators satisfy the time-inhomogeneous semigroup composition law
\begin{align}
  &  \Lambda_{t,s}  = \Lambda_{t,u} \circ \Lambda_{u,s} \ ,\ \ \ \ \forall\ t\ge u \ge s\ge 0\ .
\label{composition}
\end{align}
We are then in a position to distinguish between different situations.
\begin{enumerate}
\item If the additional conditions
\begin{align}
    \gamma_{\ell,t} \geq 0\qquad\forall \ \ell=1,\dots\,d^{2}-1\ ,
    \label{ratefun}
\end{align}
hold true for all $t \ge 0$, i.e., $\mathfrak{L}_t$ is conditionally completely positive for all $t \geq 0$, then all the maps (\ref{Texp}) are also CP  \cite{RiHu2012,RiHuPl2010}.
In such a case one calls the dynamical map $\{\Lambda_{t,0}\}_{t\geq 0}$ CP-divisible \cite{RiHuPl2010,NM1}.
Note that (\ref{Texp}) implies
\begin{equation}   \label{cpd}
    \Lambda_{t,0} = \Lambda_{t,t_n} \circ \Lambda_{t_n,t_{n-1}} \circ \ldots \circ \Lambda_{t_1,0} ,
\end{equation}
for an arbitrary set $\{t_n,t_{n-1},\ldots,t_1\}$ of times satisfying $t \geq t_n \geq t_{n-1} \geq \ldots \geq t_1 \geq 0$. Owing to the probabilistic interpretation of quantum-dynamical maps, CP-divisibility is reminiscent of the Chapman-Kolmogorov equation necessarily satisfied by transition functions of classical Markov processes (see, e.g., \cite{Kampen2007}, \cite[\S~2.2]{PavG2014}). This analogy and the observation that a time-homogeneous semigroup can always be thought of as a particular case of a time-inhomogeneous one satisfying
\begin{align}
&\Lambda_{t,s}=\Lambda_{t-s,0} = e^{(t-s)\mathfrak{L}} \qquad\forall\  t-s\ge s\ge 0
\nonumber
\end{align}
motivates us to say that {\em Markovian} means any CP-divisible dynamical map.  It should, however, be emphasized that in the literature the adjectives Markovian and non-Markovian are also used with other meanings \cite{NM1,NM2,NM3,ChrD2022,WoEiCuCi2008}. In particular, the review \cite{NM4} shows an intricate hierarchy of various and context-dependent notions.
\item When the conditions (\ref{ratefun}) do not hold, we say that the evolution represented by $\{\Lambda_{t,0}\}_{t\geq 0}$  is non-Markovian. However, the propagators---which are no longer CP---may still have additional properties. In particular, if $\mathfrak{L}_t$ is conditionally $k$-positive for all $t\ge 0$, then each propagator $\Lambda_{t,s}$ is $k$-positive for all $t \ge s \geq 0$.  One calls $\{\Lambda_{t,0}\}_{t \geq 0}$ $k$-divisible \cite{DC+Sabrina} if all propagators are $k$-positive and trace-preserving.  Now, let $\lambda_{\ell,t}$ be the time-dependent eigenvalues of $\mathfrak{L}_t$ and $\Gamma_{\ell,t} := -{\rm Re}\, \lambda_{\ell,t}$ the corresponding time-local relaxation rates. We arrive at the following
\begin{corollary} If the dynamical map $\{\Lambda_{t,0}\}_{t\geq 0}$ is 2-divisible, then
\begin{equation}   \label{GGt}
  \Gamma_{k,t} \leq \frac 1d \sum_{\ell=1}^{d^2-1} \Gamma_{\ell,t}\, ,
\end{equation}
for all $k=1,\ldots,d^2-1$ and $t \geq 0$.  In particular (\ref{GGt}) holds true for any Markovian evolution.
\end{corollary}
Let us recall that $k$-divisibility can be characterized in terms of monotonicity of the trace-norm \cite{DC+Sabrina}, that is, $\Lambda_{t,0}$ is $k$-divisible if
\begin{equation}
    \frac{d}{dt} \| ({\rm id}_k \otimes \Lambda_{t,0})(\mathbb{X})\|_1 \leq 0 ,
\end{equation}
for all Hermitian block operators $\mathbb{X} \in \mathcal{M}_k(\Md)$ ($\|\,.\,\|_1$ denotes the usual trace-norm).
\begin{corollary}
    If  $\frac{d}{dt} \| ({\rm id}_2 \otimes \Lambda_{t,0})(\mathbb{X})\|_1 \leq 0$
for all Hermitian block operators $\mathbb{X} \in \mathcal{M}_2(\Md)$, then (\ref{GGt}) holds.
\end{corollary}
\item Consider now the  master equation for an observable $X_t$
\begin{equation}
    \dot{X}_t = \mathfrak{L}^\dagger_t(X_t) = \imath\,\left[\operatorname{H}_t,X_t\right]+
\sum_{\ell=1}^{d^{2}-1}\gamma_{\ell,t} \left( \operatorname{L}_{\ell,t}^\dagger X_t \operatorname{L}_{\ell,t} - \frac 12 \{ \operatorname{L}_{\ell,t}^\dagger \operatorname{L}_{\ell,t},X_t\} \right) .
\end{equation}
\begin{corollary} If $\mathfrak{L}_t^\dagger$ is dissipative, i.e., if it satisfies \eqref{L!}
for all $X \in \Md$ and $t \ge 0$, then the time-dependent version of (\ref{Sd}) holds true, i.e.,
\begin{equation}   \label{Sdt}
  \Gamma_{k,t} \leq \frac{2}{d+1} \sum_{\ell=1}^{d^2-1} \Gamma_{\ell,t}\, ,
\end{equation}
for all $k=1,\ldots,d^2-1$ and $t \geq 0$.
\end{corollary}
In this case the propagators $\Lambda_{t,s}^\dagger$ are unital Schwarz maps for $t \geq s \ge 0$ (such dynamics were called Schwarz-divisible in \cite{DC-Farrukh}).
\end{enumerate}
\begin{example}
    Consider the well-known qubit generator \cite{ChrD2022,BrPe2002,HaCrLiAn2014}
\begin{equation}
    \mathfrak{L}_t(\bm{\rho}) = - \imath \frac{\epsilon_t}{2}[\sigma_z,\bm{\rho}] + \gamma_{+,t} \left( \sigma_+ \bm{\rho}\sigma_- - \frac 12 \{\sigma_-\sigma_+, \bm{\rho}\} \right) + \gamma_{-,t} \left( \sigma_- \bm{\rho}\sigma_+ - \frac 12 \{\sigma_+\sigma_-, \bm{\rho}\} \right) + \gamma_{z,t}(\sigma_z \bm{\rho}\sigma_z - \bm{\rho}) ,
\end{equation}
where $\sigma_\pm = \frac 12 (\sigma_x \pm \imath \sigma_y)$. The corresponding time-local relaxation rates (longitudinal and transversal) read
$$    \Gamma_{{\rm L},t} = \gamma_{+,t} + \gamma_{-,t} \ , \ \ \   \Gamma_{{\rm T},t} = \frac{\gamma_{+,t} + \gamma_{-,t}}{2} + 2 \gamma_{z,t}  \ , $$
and the transversal rate is doubly degenerated. If  $\gamma_{\pm,t}\geq 0$ and $\gamma_{z,t} \geq 0$ then Markovianity (in this case 2-divisibility) implies the standard relation
\begin{equation}   \label{2GG}
    2 \Gamma_{{\rm T},t} \geq \Gamma_{{\rm L},t} .
\end{equation}
Following \cite{HaCrLiAn2014} consider $\gamma_{\pm,t}=1$ and $\gamma_{z,t} = - \mu\, {\rm tanh}t$, that is, if $\mu > 0$ one of the rate is negative for $t > 0$. It turns out that such a generator still gives rise to a CPTP dynamical map if $\mu \leq \frac 12$ \cite{HaCrLiAn2014}. In that case $\Gamma_{{\rm L},t}=2$ and $\Gamma_{{\rm T},t}=1- 2\mu\, {\rm tanh}t$, and hence (\ref{2GG}) is violated for all $t > 0$. However, for $\mu \leq \frac 14$ the evolution is Schwarz-divisible, i.e., all propagators $\Lambda_{t,s}$ are unital Schwarz maps for $t \geq s \geq 0$.
\end{example}
}

\section{Conclusions}   \label{S-C}

In this work, we have derived a family of universal constraints for the relaxation rates of quantum-dynamical semigroups by extending the previously known constraint for GKLS generators to broader classes of maps.
We presented a new, purely algebraic derivation of the inequality
\begin{equation} \label{A}
 \Gamma_{\rm max} \leq \frac 1d \sum_{\ell=1}^{d^2-1} \Gamma_\ell \,,
\end{equation}
which holds for any GKLS generator of a completely positive, trace-preserving semigroup of a
$d$-level quantum system. Unlike an earlier proof based on classical Lyapunov theory \cite{JPA2025}, our approach relies solely on the fact that the corresponding generator is conditionally completely positive. Key elements of the proof include the use of the KMS inner product and the application of the Bendixson--Hirsch inequality.
Recall that the KMS inner product provides an important tool for the analysis of quantum detailed balance \cite{Fagnola2010,AmCa2021}.
Interestingly, our proof demonstrated that complete positivity is not essential for the validity of this bound and it suffices for the generator to be conditionally 2-positive, i.e., it is sufficient that $\mathfrak{L}$ gives rise to a 2-positive trace-preserving semigroup. For qubits, 2-positivity and complete positivity coincide, but for higher-dimensional systems, this distinction is crucial.

By further relaxing positivity assumptions to semigroups of unital Schwarz maps (a weaker condition than 2-positivity but stronger than positivity), we derived a modified bound
\begin{equation} \label{B}
 \Gamma_{\rm max} \leq \frac{2}{d+1} \sum_{\ell=1}^{d^2-1} \Gamma_\ell \,,
\end{equation}
which recovers a known qubit-case result \cite{CKM24}. For maps which are merely positive, only the trivial bound $ \Gamma_{\rm max} \leq \sum_{\ell} \Gamma_\ell $ holds, thus illustrating a clear hierarchy regarding the strength of constraints depending on the degree of positivity.

The derived constraints (\ref{A}) and (\ref{B}) also imply upper bounds on the number of steady states for the generators, with the tightest bounds corresponding to higher degrees of positivity.
For completely positive dynamics the bound `$d^2-2d+2$' was recently derived in \cite{Facchi-Amato}. Here we found a weaker bound `$d^2-d$' for 2-positive maps, the even weaker bound `$d^2-\frac 12 (d+1)$' for Schwarz maps.

Going one step further, if the generator $\mathfrak{L}_t$ is time-dependent, then $\Lambda_{t,0}$ is no longer a semigroup, but the dynamical map is instead defined via the time-ordered exponential $\Lambda_{t,0} = \mathcal{T} \exp\big(\int_0^t \mathfrak{L}_\tau d\tau\big)$, where $\mathcal{T}$ denotes a chronological product. If, for any $t\geq 0$, $\mathfrak{L}_t$ is a legitimate GKLS generator, then $\Lambda_{t,0}$ is not only completely positive but all intermediate maps (propagators)
$\Lambda_{t,s} = \mathcal{T} \exp\big(\int_s^t \mathfrak{L}_\tau d\tau\big)$ are completely positive for $t \geq s$. One calls such maps CP-divisible and the corresponding evolution Markovian \cite{NM1,NM2,NM3,NM4,ChrD2022}. It is, therefore, clear that for Markovian evolutions the rates $\Gamma_{k,t}$---which are now time-dependent---the constraint (\ref{A}) is still valid for all $t\geq 0$. Note, however, that even if CP-divisibility is violated, i.e., the evolution is non-Markovian, but if all propagators $\Lambda_{t,s}$ are at least 2-positive, then the constraint (\ref{A}) still holds. Such dynamical maps are called 2-divisible \cite{DC+Sabrina}. Hence, our result clarifies the essential role of 2-positivity, and at the same time offers new tools for diagnosing non-Markovianity via the framework of divisibility.
{Recall, that maps which are positive but not completely positive provide a basic theoretical tool for entanglement detection \cite{HHH,HHHH}. In particular 2-positive maps can be used to detect weakly entangled states, i.e., states with Schmidt number greater than two. Hence, whenever the dynamical map $\Lambda_{t,0}$ is 2-divisible and a bipartite state $\bm{\rho}$ is of Schmidt number at most two,  $({\rm id} \otimes \Lambda_{t,s})(\bm{\rho})$ still defines a legitimate quantum state in spite of the fact that the propagators $\Lambda_{t,s}$ are only 2-positive.\\
\indent Interestingly, the authors of the two seminal papers \cite{LinG1976,GoKoSu1976} used two different approaches to derive the structure of the generator of semigroups of completely positive maps. Gorini et al.~\cite{GoKoSu1976} worked in the Schr\"odinger picture and generalized conditional positivity
\begin{equation}  \label{c1}
      \< \psi|\mathfrak{L}(|\phi\>\<\phi)|\psi\> \geq 0 \ , \ \ \ \forall\ \psi \perp \phi ,
\end{equation}
to conditional $k$-positivity, and eventually conditional complete positivity. Using the property of the Choi matrix \cite{Choi-75} finally let them derive the structure of the generator.  Lindblad \cite{LinG1976} worked in the Heisenberg picture and generalized the dissipativity condition
\begin{equation}  \label{c2}
     \mathfrak{L}^\dagger(X^\dagger X) \geq \mathfrak{L}^\dagger(X^\dagger) X + X^\dagger \mathfrak{L}^\dagger(X) \,,
\end{equation}
to $k$-dissipativity, and eventually complete dissipativity.  Note, that (\ref{c2}) is stronger than (\ref{c1}) (unless we restrict (\ref{c2}) to hold only for Hermitian operators in which case they are equivalent). However, conditional complete positivity and complete dissipativity do coincide. Our results clearly show that already conditional 2-positivity and dissipativity give rise to different but universal constraints for the corresponding relaxation rates. This way the above---quite abstract mathematical---properties turned out to imply clear distinctions for physically measurable quantities.  Therefore, our approach can be seen as reminiscent of Bell’s theorem, which enabled direct experimental tests of assumptions about realism and locality.
}

Finally, let us address the elephant in the room: since the bound (\ref{CON}) holds for all 2-positive trace-preserving semigroups, and since it is also tight for completely positive dynamics (i.e., it cannot be improved further), how can one distinguish these two properties, spectrally?
In other words, are there structural differences between the spectrum of generators of 2-positive and generators of completely positive dynamics?
As explained just now, if such a simple spectral constraint (separating $2$-positivity and complete positivity) exists, it cannot be of the form (\ref{CON}), but it, by definition, still had to be expressible through the relaxation rates $\Gamma_\ell$ of $\mathfrak{L}$. {One way would be to look for nonlinear constraints. Another option might be to consider linear constraints but, e.g., for sums of say $m \geq 2$ maximal rates. }
Another important issue is how to generalize the constraints (\ref{A}) and (\ref{B}) for the infinite dimensional scenario. It would be interesting to investigate this problem for at least some classes of systems starting from, e.g., Gaussian semigroups.
We leave these challenging yet highly interesting questions as directions for future research.

\appendix

\section{The generators with faithful steady states are dense}\label{app_b}

The following result is not surprising but also not trivial, which is why we decided to prove it explicitly.
\begin{proposition}\label{prop_kpos_density}
 For any $k\in\mathbb N$ define ${\rm cP}_k$ (conditionally $k$-positive) as the set of all linear maps $\mathfrak{L}: \Md \to \Md$ such that $e^{t\mathfrak{L}}$ is $k$-positive and trace-preserving for all $t\geq 0$.
 Moreover, define ${\rm cP}_{k,+}$ as the set of all $\mathfrak{L}\in{\rm cP}_{k}$ such that $\mathfrak{L}(\omega)=0$ for some $\omega>0$. Then $\overline{{\rm cP}_{k,+}}= {\rm cP}_{k}$.
\end{proposition}
\begin{proof}
 $\subseteq$: Obvious, because ${\rm cP}_{k,+}\subseteq{\rm cP}_{k}$ by definition and ${\rm cP}_{k}$ is closed \cite[Prop.~1.14]{LNM1552}.
 $\supseteq$: Let $\mathfrak{L}\in{\rm cP}_k$. Define
 \begin{equation}
 \mathfrak{L}_0(\bm{\rho}) = \frac 1d \oper_d {\rm Tr}\bm{\rho} - \bm{\rho}\,.
 \end{equation}
 The idea will be to show that for any $n\in\mathbb N$ the generator $\mathfrak{L} + \frac 1n \mathfrak{L}_0$ is actually in ${\rm cP}_{k,+}$, which would conclude the proof because then
\begin{equation}
 \mathfrak{L} = \lim_{n\to\infty} \mathfrak{L} + \frac1n \mathfrak{L}_0 \in \overline{{\rm cP}_{k,+}}\,.
\end{equation}
 We proceed in 
 two steps.

 \begin{enumerate}
 \item[Step 1:] {First, recall Duhamel's integral formula \cite[Ch.~1, Thm.~5.1]{DF84}, \footnote{We thank one of the anonymous referees for sharing with us how Duhamel's formula can be used to substantially simplify the proof.}
\begin{equation}\label{eq:duhamel}
    e^{tA} - e^{tB} = \int_0^t ds\, e^{sB} (A-B) e^{(t-s)A} \,.
\end{equation}
Choosing $A = \mathfrak{L} + \frac 1n \mathfrak{L}_0 $ and $B= \mathfrak{L} - \frac 1n {\rm id}_d $ one finds $A-B=\frac1{nd}\oper_d {\rm Tr}\bm{\rho} $ and thus
$$ (A- B)e^{(t-s)A}(\bm{\rho}) =  \frac 1{nd}\oper_d {\rm Tr}\big( e^{(t-s)A} (\bm{\rho}) \big)=  \frac 1{nd}\oper_d {\rm Tr}\bm{\rho} , $$
for any $\bm{\rho}\in \Md$.
Inserting this into~\eqref{eq:duhamel} yields
\begin{equation}
    e^{t(\mathfrak{L} + \frac 1n \mathfrak{L}_0)}(\bm{\rho}) =  e^{-\frac tn}e^{t\mathfrak{L}} (\bm{\rho}) +{\rm Tr}\bm{\rho} \; G(t) ,
\end{equation}
with
\begin{equation}
    G(t) = \frac 1n \int_0^t ds\, e^{- \frac sn} e^{s \mathfrak{L}} \left(\frac 1d \oper_d\right) .
\end{equation}   }

Note that $G(t)$ is an analytic curve of Hermitian matrices which is independent of $\bm{\rho}$.
{Now, since the set of full-rank matrices is open, by
continuity there exists $\epsilon>0$ such that $e^{s \mathfrak{L}} \left(\frac 1d \oper_d\right)$ is full rank for all $s \in (0, \epsilon)$. Hence $G(\epsilon) > 0$ and, in particular,}
\begin{equation}
 e^{t( \mathfrak{L} + \frac 1n \mathfrak{L}_0 )}(\bm{\rho}) > 0
\end{equation}
 for all $t\in(0,\varepsilon)$ and all non-zero $\bm{\rho}\geq 0$.

 \item[Step 2:] Because $\mathfrak{L}':= \mathfrak{L} + \frac 1n \mathfrak{L}_0\in {\rm cP}_k$, the map $e^{\frac{\varepsilon}{2}\mathfrak{L}'}$ in particular is positive and trace-preserving. Hence by the Brouwer fixed point theorem \cite{Brouwer11} there exists a state $\bm{\sigma}$ such that
 $ e^{\frac{\varepsilon}{2}\mathfrak{L}'}(\bm{\sigma})= \bm{\sigma} $.
 Using this we define
 $$
 \bm{\sigma}':=\int_0^{\frac{\varepsilon}{2}}e^{t\mathfrak{L}'}(\bm{\sigma})\,dt \,,
 $$
 and we claim that $\bm{\sigma'}>0$ and, more importantly, $\mathfrak{L}'(\bm{\sigma}')=0$. This would show that $\mathfrak{L}'\in {\rm cP}_{k,+}$ for all $n$, as desired.
 For the proof of this we follow the argument of \cite[Prop.~5]{BNT08b}:
 \begin{eqnarray*}
 \mathfrak{L}'(\bm{\sigma}') =\int_0^{\frac{\varepsilon}{2}}\mathfrak{L}'e^{t\mathfrak{L}'}(\bm{\sigma})\,dt = \Big(\int_0^{\frac{\varepsilon}{2}}\mathfrak{L}'e^{t\mathfrak{L}'}\,dt\Big)(\bm{\sigma}) =\big(e^{\frac{\varepsilon}{2}\mathfrak{L}'}-{\rm id}\big)(\bm{\sigma})=0\,.
 \end{eqnarray*}
 Moreover, $\bm{\sigma}'>0$ because for all vectors $x\neq 0$, $t\mapsto\langle x|e^{t\mathfrak{L}'}(\bm{\sigma})|x\rangle$ is a continuous, non-negative curve which is strictly positive at $t=\frac{\varepsilon}{2}$ (because $e^{\frac{\varepsilon}{2}\mathfrak{L}'}(\bm{\sigma})>0$ by our previous argument). Hence
 $$
 \langle x|\bm{\sigma}'|x\rangle=\int_0^{\frac{\varepsilon}{2}}\langle x|e^{t\mathfrak{L}'}(\bm{\sigma})|x\rangle\,dt>0 \,,
 $$
 which concludes the proof.\qedhere
 \end{enumerate}
\end{proof}

\section{Proof of Proposition 6}\label{sec_app_proof_pro_II}

Let us start by fixing two orthonormal bases $\{|1\>,|2\>\}$ in $\mathbb{C}^2 $ and $\{|e_1\rangle,\ldots,|e_d\rangle\}$ in $\mathbb{C}^d$. With this we define
\begin{equation}\label{eq:def_phi_pm}
 |\phi^\pm_{ij}\> = |1\> \otimes |e_i\> \pm |2\> \otimes |e_j\> \in \mathbb{C}^2 \otimes \mathbb{C}^d\,.
\end{equation}
Clearly for all $i \neq j$ one has $|\phi^+_{ij}\> \perp |\phi^-_{ij}\>$ and hence
\begin{equation}
 \< \phi^-_{ij} |\, [{\rm id}_2 \otimes \mathfrak{L}](|\phi^+_{ij}\>\<\phi^+_{ij}|)\, |\phi^-_{ij}\> \geq 0 \,,
\end{equation}
for any pair $i \neq j$ due to the fact that $\mathfrak{L}$ is conditionally 2-positive. Expanding this using (\ref{eq:def_phi_pm}) yields
\begin{equation*}
 [{\rm id}_2 \otimes \mathfrak{L}](|\phi^+_{ij}\>\<\phi^+_{ij}|) = |1\>\<1| \otimes \mathfrak{L}(|e_i\>\<e_i|) +
 |1\>\<2| \otimes \mathfrak{L}(|e_i\>\<e_j|) + |2\>\<1|\otimes \mathfrak{L}(|e_j\>\<e_i|) + |2\>\<2| \otimes \mathfrak{L}(|e_j\>\<e_j|) \,,
\end{equation*}
and hence
\begin{equation*}
 \< \phi^-_{ij} |\, [{\rm id}_2 \otimes \mathfrak{L}](|\phi^+_{ij}\>\<\phi^+_{ij}|)\, |\phi^-_{ij}\> = \<e_i|\mathfrak{L}(|e_i\>\<e_i|)|e_i\> + \<e_j|\mathfrak{L}(|e_j\>\<e_j|)|e_j\> - \<e_i|\mathfrak{L}(|e_i\>\<e_j|)|e_j\> - \<e_j|\mathfrak{L}(|e_j\>\<e_i|)|e_i\>\,.
\end{equation*}
Finally, recalling that
\begin{equation} \label{B3}
 {\rm Tr}\, \mathfrak{L} = \sum_{i,j=1}^d \<e_i|\mathfrak{L}(|e_i\>\<e_j|)|e_j\> = {\rm Tr}\, \mathcal{K} + \sum_{i\neq j} \<e_i|\mathfrak{L}(|e_i\>\<e_j|)|e_j\> \,,
\end{equation}
one computes
\begin{eqnarray}
 \sum_{i\neq j} \< \phi^-_{ij} |\, [{\rm id}_2 \otimes \mathfrak{L}](|\phi^+_{ij}\>\<\phi^+_{ij}|)\, |\phi^-_{ij}\> &=& \sum_{i\neq j} \left( \<e_i|\mathfrak{L}(|e_i\>\<e_i|)|e_i\> + \<e_j|\mathfrak{L}(|e_j\>\<e_j|)|e_j\> \right) \nonumber \\ &-& \sum_{i\neq j} \left( \<e_i|\mathfrak{L}(|e_i\>\<e_j|)|e_j\> + \<e_j|\mathfrak{L}(|e_j\>\<e_i|)|e_i\> \right)\,.
\end{eqnarray}
Now,
\begin{equation}
 \sum_{i\neq j} \left( \<e_i|\mathfrak{L}(|e_i\>\<e_i|)|e_i\> + \<e_j|\mathfrak{L}(|e_j\>\<e_j|)|e_j\> \right) = 2(d-1)\, {\rm Tr}\, \mathcal{K} \,,
\end{equation}
and due to (\ref{B3})
\begin{equation} \label{BLK}
 \sum_{i\neq j} \left( \<e_i|\mathfrak{L}(|e_i\>\<e_j|)|e_j\> + \<e_j|\mathfrak{L}(|e_j\>\<e_i|)|e_i\> \right) = 2({\rm Tr}\, \mathfrak{L} - {\rm Tr}\, \mathcal{K}) \,,
\end{equation}
which implies
\begin{eqnarray}
 \sum_{i\neq j} \< \phi^-_{ij} |\, [{\rm id}_2 \otimes \mathfrak{L}](|\phi^+_{ij}\>\<\phi^+_{ij}|)\, |\phi^-_{ij}\> = 2(d\,{\rm Tr}\, \mathcal{K} - {\rm Tr}\, \mathfrak{L} ) \geq 0 \,,
\end{eqnarray}
and finally proves Proposition 6.

\section{Proof of Proposition 10} \label{sec_app_proof_pro_III}

The dissipativity condition

$$  \mathfrak{L}^\dagger(X^\dagger X) \geq \mathfrak{L}^\dagger(X^\dagger) X + X^\dagger \mathfrak{L}^\dagger(X) \,, $$
applied to $X = |e_i\>\<e_j|$ gives
\begin{equation}
 \mathfrak{L}^\dagger(|e_j\>\<e_j|) \geq \mathfrak{L}^\dagger(|e_j\>\<e_i|) |e_i\>\<e_j| + |e_j\>\<e_i| \mathfrak{L}^\dagger(|e_i\>\<e_j|) \,,
\end{equation}
and hence
\begin{equation}
 \mathcal{K}_{jj} \equiv \<e_j|\, \mathfrak{L}^\dagger(|e_j\>\<e_j|)\,|e_j\> \geq \<e_j|\,\mathfrak{L}^\dagger(|e_j\>\<e_i|)\, |e_i\> + \<e_i|\, \mathfrak{L}^\dagger(|e_i\>\<e_j|)\,| e_j\>\,.
\end{equation}
This implies
\begin{equation}
 \sum_{i\neq j} \mathcal{K}_{jj} \geq \sum_{i\neq j} \left( \<e_j|\,\mathfrak{L}^\dagger(|e_j\>\<e_i|)\, |e_i\> + \<e_i|\, \mathfrak{L}^\dagger(|e_i\>\<e_j|)\,| e_j\> \right) \equiv 2({\rm Tr}\, \mathfrak{L} - {\rm Tr}\, \mathcal{K}) \,,
\end{equation}
where we used (\ref{BLK}). Altogether, one computes
\begin{equation}
 \sum_{i\neq j} \mathcal{K}_{jj} = (d-1)\, {\rm Tr}\, \mathcal{K} \,,
\end{equation}
and hence
\begin{equation}
 (d-1) \, {\rm Tr}\, \mathcal{K} \geq 2({\rm Tr}\, \mathfrak{L} - {\rm Tr}\, \mathcal{K}) \,,
\end{equation}
which proves Proposition 10.

\section*{Acknowledgments}
	
DC was supported by the Polish National Science Center
under Project No. 2018/30/A/ST2/00837.
FvE is funded by the \textit{Deutsche Forschungsgemeinschaft} (DFG, German Research Foundation) -- project number 384846402, and supported by the Einstein Foundation (Einstein Research Unit on Quantum Devices) and the MATH+ Cluster of Excellence.
G.K. was supported by JSPS KAKENHI Grants Nos. 24K06873. PMG acknowledges support by the CoE in Randomness and Structures (FiRST) of the Research Council of Finland (funding decision number: 346305). {We thank the anonymous referees for several valuable comments and remarks which enabled us to significantly improve the paper's presentation. In particular, we are grateful to one of the referees for suggesting the use of Duhamel's integral formula to substantially simplify the proof
of the fact that the generators with faithful steady states are dense.}

\vspace{.2cm}

	\bibliography{LGKS_TM}{} 
	\bibliographystyle{apsrev4-2}
\end{document}